%% file: main.tex
\documentclass[11pt]{amsart}
\usepackage[
  margin=1.1in
]{geometry}
\usepackage{amsmath,amssymb,amsthm,mathtools}
\usepackage{enumitem}
\usepackage{xcolor}
\usepackage[colorlinks=true,linkcolor=blue!60!black,citecolor=blue!60!black,urlcolor=blue!60!black]{hyperref}
\makeatletter
\renewcommand{\paragraph}{\@startsection{paragraph}{4}%
  \z@{.5\linespacing\@plus.7\linespacing}{-\fontdimen2\font}%
  {\normalfont\bfseries}}
\makeatother

\newtheorem{theorem}{Theorem}[section]
\newtheorem{lemma}[theorem]{Lemma}

\newtheorem{corollary}[theorem]{Corollary}

\theoremstyle{definition}

\theoremstyle{remark}
\newtheorem{remark}[theorem]{Remark}

\newcommand{\E}{\mathbf E}

\newcommand{\poly}{\operatorname{poly}}

\newcommand{\ind}{\mathbf{1}}

\DeclareMathOperator{\MSB}{MSB}
\DeclareMathOperator{\Par}{Par}

\usepackage[
backend=biber,
style=alphabetic,
sorting=nyt,
maxbibnames=99,
giveninits=true,
doi=true,
url=true,
eprint=true
]{biblatex}

\title[Truly work-efficient deterministic coloring and MIS]
{Truly work-efficient parallel deterministic
$(\Delta+1)$-coloring and maximal independent set}
\author[C. Hutton and A. Melrod]{Chase Hutton}
\thanks{$^{*}$University of Maryland, College Park, MD}
\thanks{$^{\dagger}$Harvard University, Cambridge, MA}
\makeatletter
\renewcommand{\@setauthors}{%
  \begingroup
  \medskip
  \centering
  \begin{tabular}{c@{\hspace{4em}}c}
    Chase Hutton$^{*}$ & Adam Melrod$^{\dagger}$ \\[2pt]
    \texttt{chutton6@umd.edu} & \texttt{amelrod@math.harvard.edu}
  \end{tabular}\par
  \endgroup
}
\makeatother

\begin{document}
\begin{abstract}
      We give deterministic parallel algorithms that compute a $(\Delta+1)$-coloring and a maximal independent set for a simple graph with $n$ vertices and $m$ edges in $O(n+m)$ work and $O(\poly\log n)$ depth.
\end{abstract}
\maketitle
\input{short-note/intro}
\input{short-note/prelims}
\input{short-note/benefit-problems}
\input{short-note/graph-sparsification}
\input{short-note/deterministic-graph-algos}

\section{Acknowledgements and Declarations}
\paragraph{Acknowledgements.}

The authors thank the anonymous reviewers for their helpful feedback and remarks that simplified the recursive bisection procedure of Section \ref{sec:sparsification}.

\paragraph{AI Disclosure.}

AI was used to assist with the literature review. AI was not used for any substantive idea contribution. The authors take full responsibility for the correctness and originality
of this work.

\printbibliography[
heading=bibintoc,
title={References}
]

\end{document}

%% file: short-note/intro.tex
\section{Introduction}\label{sec:intro}

In the maximal independent set (MIS) problem, we are given a graph $G$ with $n$ vertices and $m$ edges, and the task is to output a set $M \subseteq V(G)$ such that no two vertices of $M$ are neighbors but every vertex outside of $M$ has a neighbor in $M$. In the $(\Delta+1)$-coloring problem, the task is to color every vertex by a number in $\{0,\dots,\Delta\}$, where $\Delta$ is the maximum degree of $G$, so that no edge is monochromatic. Sequentially, both problems are solved by the same greedy procedure in $O(n+m)$ time: the vertices are processed in an arbitrary order, and the solution is extended at each vertex in accordance with its previously processed neighbors. The question of matching this trivial linear bound in parallel has driven research in parallel computing since the early 1980s \cite{KarpWigderson1985MIS,Luby1986MIS,AlonBabaiItai1986MIS}.

To make the question precise, we adopt the standard work--depth model \cite{JaJa1992}. The \emph{work} of a parallel algorithm is
the total number of operations it performs, and its \emph{depth} is the
length of its longest chain of sequentially dependent operations. A parallel algorithm should therefore have small depth and, more importantly, it should not perform substantially more work than its sequential counterpart: an
algorithm with superlinear work is slower than the trivial sequential
algorithm unless the number of available processors exceeds the work
overhead. Accordingly, we call an algorithm for either problem \emph{work-efficient} if its work is $O(n+m)$, matching the sequential greedy bound, and its depth is polylogarithmic in $n$.

For randomized algorithms the question has largely been resolved. The
classical algorithms of Luby \cite{Luby1986MIS} and of Alon, Babai, and Itai
\cite{AlonBabaiItai1986MIS} cost linear work and run in $O(\log n)$ rounds.
Blelloch, Fineman, and Shun \cite{BlellochFinemanShun2012} later showed that
the sequential greedy algorithm for MIS, applied to a uniformly random
vertex order, admits a direct parallelization. Notably, the parallelization
performs $O(n+m)$ work regardless of the order, and only its depth
guarantee, logarithmic with high probability \cite{tight-mis}, depends on
the randomness. For deterministic algorithms, which are the subject of this
paper, the question has proved far more resistant.

The natural route to a deterministic parallel algorithm is to derandomize a
randomized one, and parallel derandomization is a classical subject in its
own right. Karp and Wigderson \cite{KarpWigderson1985MIS} initiated it with
a deterministic MIS algorithm of polylogarithmic depth and
$\widetilde{O}(n^3)$ work. Luby \cite{Luby1986MIS} showed that his
randomized algorithm can be derandomized by exhaustively searching a
pairwise-independent sample space of quadratic size, giving $O(\log^2 n)$
time on $O(mn^2)$ processors, and later removed this penalty through the
method of conditional probabilities \cite{Luby1993RemovingRandomness}.
General-purpose derandomization techniques followed
\cite{MotwaniNaorNaor1994,BergerRompel1991}, as did sharper algorithms for
MIS: Goldberg and Spencer
\cite{GoldbergSpencer1989Constructing,GoldbergSpencer1989NewParallel}
brought the work to $O\big((n+m)\log^2 n\big)$ in polylogarithmic depth,
and Han \cite{Han1996Derandomization} computes an
MIS in $O(\log^{2.5} n)$ time on $O\big((m+n)/\log^{1.5} n\big)$
processors, for a total of $O\big((n+m)\log n\big)$ work. Most recently, Ghaffari, Grunau, and
Rozho\v{n} developed a program of parallel derandomization
\cite{GhaffariGrunauRozhon2023DerandomizationI,GhaffariGrunau2024DerandomizationII} of Chernoff-like bounds, that for example is applicable to edge coloring.

\paragraph{State of the Art.} Following their contributions to derandomizing Chernoff-like bounds, Ghaffari and Grunau \cite{GhaffariGrunau2025TrueWorkEfficiency} developed derandomized algorithms with $(n+m)\poly(\log\log n)$ work and polylogarithmic depth for MIS and maximal matching. Their algorithms derandomize a hitting set abstraction that captures maximal matching and, in a more involved form, MIS, and that has the potential to apply to a broader class of problems. The method rounds fractional assignments to integral ones over $\Theta(\log n)$ stages whose accumulated errors are governed by an extremely intricate potential analysis. Moreover, the application of their method to the MIS problem requires a significant amount of additional work. Ghaffari and Grunau pose explicitly the question of whether the remaining $\poly(\log\log n)$ overhead can be removed and true work efficiency reached.

For $(\Delta+1)$-coloring, the best previously known deterministic bound is the $O\big((n+m)\log^2 n)$ work algorithm given by Han
\cite{Han1996Derandomization}, building off of Luby's work providing an $O\big((n+m)\log^3 n \log \log n\big)$ work solution \cite{Luby1993RemovingRandomness}. Prior to this work, no better algorithm had appeared in the three decades since.

Two further lines of work deserve mention. When the maximum
degree of the graph is constant, work-efficient deterministic algorithms for both problems have long been known through the deterministic coin tossing
technique of Cole and Vishkin
\cite{ColeVishkin1986CoinTossing,GoldbergPlotkinShannon1988}. In the distributed and massively parallel settings, deterministic round-efficient algorithms for coloring and MIS have advanced considerably
\cite{RozhonGhaffari2020,GhaffariKuhn2021Coloring,GhaffariGrunau2024NetworkDecomposition,CoyCzumajDaviesPeckMishra2024}, but simulating a round-efficient message-passing algorithm is inherently inefficient.

On general graphs, among the three classical
symmetry-breaking problems, only maximal matching was known to admit a
work-efficient deterministic parallel algorithm
\cite{Kelsen1994MaximalMatching,han1995improvement}. Whether MIS or
$(\Delta+1)$-coloring admits one had remained open.

\paragraph{Our Contributions.}
This paper answers the question in the affirmative for both problems.

\begin{theorem}[cf.\ Corollary~\ref{cor:mis-linear}]\label{thm:intro-mis}
A maximal independent set of any graph can be computed deterministically in
$O(n+m)$ work and $\poly(\log n)$ depth.
\end{theorem}

\begin{theorem}[cf.\ Corollary~\ref{cor:coloring-linear}]\label{thm:intro-coloring}
A $({\deg}+1)$-coloring of any graph---and in particular a
$(\Delta+1)$-coloring---can be computed deterministically in $O(n+m)$ work
and $\poly(\log n)$ depth.
\end{theorem}

These are the first work-efficient deterministic parallel algorithms for either problem. Independently and concurrently, Ghaffari and Jiang obtained similar results by different methods \cite{GhaffariJiang2026TrueWorkEfficiency} (see the end of this section). Theorem~\ref{thm:intro-coloring} improves the previously best known deterministic work
bound for $(\Delta+1)$-coloring by a factor of $\log^2 n$, Theorem~\ref{thm:intro-mis} removes the final $\poly(\log\log n)$ overhead of \cite{GhaffariGrunau2025TrueWorkEfficiency}, and together with the maximal matching algorithms of Kelsen and Han they complete the picture for the three classical symmetry-breaking problems. The complete argument, including all proofs, is elementary and short, and it employs no derandomization tool beyond the pairwise-independence framework of Luby \cite{Luby1993RemovingRandomness} in the linear-work form given to it by Han \cite{han1995improvement}.

Our approach rests on the observation that both MIS and $({\deg}+1)$-coloring
are extendable: a partial solution determines a residual instance on
any set of unprocessed vertices at the cost of inspecting each edge joining
that set to the processed part exactly once---for MIS by deleting the
vertices that neighbor the current independent set, and for coloring by
removing from each list the colors of processed neighbors. A solution may
therefore be assembled piece by piece, exactly as in the sequential greedy
procedure, at a cost of $O(n+m)$ for all extension steps.

The computation on these pieces can be made cheap deterministically by an edge-sparsifying partitioning
routine that constitutes the technical core of the paper.
Theorem~\ref{cor:aggregate-sparsification} shows that, in $O(n+m)$ work and $\poly(\log n)$ depth, the vertex set of any graph can be partitioned into $L$ parts whose induced subgraphs jointly contain at most $O(m/L)$ edges. The partition is computed by recursive bisection which decreases the total internal edges by at least half.
Derandomization enters only in computing this partition: the number of internal edges cut by a random bipartition is a sum of functions of pairs of coordinates of the assignment, this is an instance of Luby's bit-pair benefit problem \cite{Luby1993RemovingRandomness}, which
Han's refinement \cite{han1995improvement} solves in $O(n+m)$ work. Importantly, since the edge counts relevant to the benefit problem decrease geometrically, this propagates to the work bound which is essential to obtain a work-efficient algorithm.

Together, these ingredients upgrade any suboptimal deterministic algorithm for an extendable problem to a work-efficient one. By taking $L = \log^{a} n$, the pieces jointly contain at most $m/\log^{a} n$ edges, so an algorithm suboptimal by a factor of $\log^{a} n$ solves all of them in $O(m)$ total work. We isolate this as
Theorem~\ref{thm:extendable-linear}. For MIS, the instances arising on the pieces are ordinary MIS instances, and any of the deterministic algorithms with $\poly\log n$ overhead discussed earlier can serve as the subroutine. For coloring, the pieces carry list instances, since colors used by earlier neighbors are removed from
the palettes, and we use the algorithm of Luby
\cite{Luby1993RemovingRandomness}, which, while presented as a $(\Delta+1)$-coloring algorithm, directly works with no modification to produce a coloring by palettes of size ${\deg} + 1$, provided they are in sorted order. This yields Theorems~\ref{thm:intro-mis} and~\ref{thm:intro-coloring}.

\paragraph{Paper Organization.}
In Section~\ref{sec:benefit-problems} we review Luby's benefit problems, and give a self-contained
linear-work solution of the bit pair benefit problem in the style of Han. In Section~\ref{sec:sparsification} we construct the deterministic sparse partition.
In Section~\ref{sec:derandom} we prove the bootstrapping theorem and derive the algorithms for MIS and $({\deg}+1)$-coloring.

\paragraph{Independent Work.}
Recently, Ghaffari and Jiang \cite{GhaffariJiang2026TrueWorkEfficiency}
independently and concurrently obtained deterministic parallel algorithms
with $O(n+m)$ work and $\poly(\log n)$ depth for MIS and
$(\Delta+1)$-coloring. Their work develops a more general framework based on gradual rounding and defective coloring and extends to other problems including maximal matching and set cover. We note that our solutions to MIS and $(\Delta+1)$-coloring are substantially simpler.

%% file: short-note/prelims.tex
\section{Preliminaries}\label{sec:prelims}

\paragraph{Model of computation.}
We work in the standard work--depth model \cite{JaJa1992,Blelloch1996}: an
algorithm runs on $p$ processors sharing a memory, its work $W(A)$ is the total number of its computational steps, and its depth $D(A)$ is the length of its longest chain of sequentially dependent steps. The
underlying machine is the \textsc{common} CRCW PRAM
\cite{KarpRamachandran1990} in which concurrent reads of a memory location are unrestricted, and concurrent writes are permitted only when all writers write the same value. We note that our maximal independent set algorithm requires no concurrent writes and runs on the weaker CREW PRAM machine while our coloring algorithm requires them only in its extension step, when the colors already used among the neighbors of a vertex are recorded. Throughout, a machine word consists of $\Theta(\log n)$ bits, where $n$ is the number of vertices of the input graph, and the standard arithmetic, comparison, and bitwise
operations on words each cost a single step. The input graph is provided as an array of $n$ vertices and each vertex $v$ possesses an array listing the neighbors of
$v$. 

\paragraph{Basic subroutines.}
All prefix sums of an array of $k$ numbers can be computed in $O(k)$
work and $O(\log k)$ depth \cite{JaJa1992}. As a consequence, an array
can be \emph{filtered} within the same bounds: the entries satisfying a
predicate testable in constant time can be moved into a contiguous
array in their original order. Sorting $k$ items under constant-time
comparisons costs $O(k \log k)$ work and $O(\log k)$ depth
\cite{Cole1988MergeSort}. We further use the following sorting
lemma.

\begin{lemma}\label{lem:small-range-sort}
    A list of $m$ items with integer keys in $[1, R]$ can be stable sorted
    deterministically in $O(m + R)$ work and $O(R + \log m)$ depth.
\end{lemma}
\begin{proof}
    Partition the items into $\lceil m/R \rceil$ blocks of at most $R$
    consecutive items. Each block, in parallel, computes sequentially
    in $O(R)$ time its key histogram and the rank of each of its items
    among equal keys within the block. The histograms form an
    $R \times \lceil m/R \rceil$ matrix of counts with $O(m + R)$
    entries, and a prefix sum over this matrix in key-major order \cite{JaJa1992} yields, for every key and block, the starting
    position of that block's items with that key in the sorted output.
    Each item then computes its output position in $O(1)$ from its
    block's starting position and its within-block rank.
\end{proof}

% \begin{lemma}[Degree sort]\label{lem:degree-sort}
%     The vertices of any graph with $n$ vertices and $m$ edges can be
%     stable sorted by degree deterministically in $O(n + m)$ work and
%     $O(\log n)$ depth.
% \end{lemma}
% \begin{proof}
%     Set $\theta := \lceil \log n \rceil$ and split the vertex list, by a filter, into the vertices of degree at most $\theta$ and the remaining vertices, each group retaining its input order. The first group is stable sorted by Lemma~\ref{lem:small-range-sort} with key
%     range $[1, \theta+1]$, at a cost of $O(n + \theta)$ work and $O(\theta + \log n) = O(\log n)$ depth. Since the degrees sum to $2m$, the second group contains at most $2m/\theta$ vertices, and it is sorted by parallel merge sort \cite{Cole1988MergeSort} at a cost of
%     $O\big((m/\theta)\log m\big) = O(m)$ work and $O(\log n)$ depth. Every degree in the first group is smaller than every degree in the second, so the
%     concatenation of the two sorted lists is the stable sorted output, and the total cost is $O(n+m)$ work and $O(\log n)$ depth.
% \end{proof}

%% file: short-note/benefit-problems.tex
\section{Benefit problems}\label{sec:benefit-problems}

The general benefit problem as formulated by Luby in \cite{Luby1993RemovingRandomness} is as follows: let $B$ be a function, called the \emph{benefit function}, from $\{0,1\}^n \to \mathbf R$. A point $x \in \{0,1\}^n$ is \emph{good} if $B(x) \geq \E_\alpha[B(\alpha)]$ where $\alpha$ is a uniformly random point of $\{0,1\}^n$, that is the benefit of $x$ is at least the average benefit. The \emph{benefit problem} for $B$ is then to efficiently find a good point. While Luby formulates the general benefit problem, he solves it in restricted cases where pairwise independence can be exploited to shrink the seed space; this is sufficient for many applications.

\subsection{Bit pair benefit problem}
The \emph{bit pair} benefit problem assumes that the benefit function can be written as a sum of auxiliary functions 
\[
    B(x) = B(x_1,\dots,x_n) = \sum_{u<v} \Psi_{uv}(x_u,x_v) + \sum_u \Psi_u(x_u) + c.
\]
Given a representation exists with at most $m$ nonzero auxiliary functions computable in constant time, that function $B$ is called a \emph{bit pair benefit function} of size $m$. This representation is exploited by observing that if the bits $x_u$ are pairwise independent and uniform, the expected value is the same as if the entire $x$ was uniformly sampled. As such, the size of the seed space can be dramatically reduced to the point where binary searching the seed space for a good point solves the problem with logarithmic overhead. Evaluating the quantities needed to do the binary search takes $O(m)$ work.
This is the approach taken by Luby in \cite{Luby1993RemovingRandomness}, which solves the benefit problem for $B$ deterministically in $O((n+m) \log^2 n)$ work and $O(\log^2 n)$ depth.

The inefficiency of this approach is primarily due to evaluating the benefit function at every step of the binary search. Han \cite{han1995improvement} later exploited, following an observation of Luby, that if one tracks the influence of each constraint more carefully, a linear algorithm can be obtained. After modifying the auxiliary functions so that they only depend on the xor of their inputs (which can be done without changing the expectation overall), each constraint is only relevant to deciding one of the bits of the seed of a good point. By grouping the relevant constraints together, redundant work is eliminated. The paper of Han only applies this observation to the problem of maximal matching, but it works more generally, without any additional difficulty in presentation. Thus, we prove the following theorem by the same approach.

\input{short-note/linear-bit-pair-problem}

%% file: short-note/linear-bit-pair-problem.tex
\begin{theorem}\label{thm:monomial-linear}
    Given a collection of constraint pairs
    $\mathcal P \subseteq \binom{[n]}{2}$ of size $m$, and
    corresponding auxiliary functions
    $\Psi_p \colon \{0,1\} \to \mathbf R$, the benefit problem for
    \[
        B(x) := \sum_{p = \{u, v\} \in \mathcal P} \Psi_p(x_u \oplus x_v)
    \]
    can be solved deterministically in $O(n + m)$ work and
    $O(\log m \log n)$ depth.
\end{theorem}

\begin{remark}[General bit pair benefit functions]\label{rem:general-bit-pair}
Theorem~\ref{thm:monomial-linear} also solves the general bit pair benefit
problem defined above with the same work and depth bounds. For $x \in \{0,1\}^n$, let $\bar x$
be its coordinatewise complement and define
\[
    \widehat B(x) := \frac{B(x)+B(\bar x)}{2}.
\]
This symmetrization makes each unary term constant and replaces each
pair term $\Psi_{uv}(a,b)$ by
$\tfrac12\big(\Psi_{uv}(a,b)+\Psi_{uv}(1-a,1-b)\big)$, which depends
only on $a \oplus b$. Thus, up to an additive constant, $\widehat B$
has the form required by Theorem~\ref{thm:monomial-linear}, with at most
$m$ constraint pairs. Moreover,
$\E_\alpha[\widehat B(\alpha)] = \E_\alpha[B(\alpha)]$.
Applying the theorem to $\widehat B$ gives a point $x$ such that
\[
    \max\{B(x),B(\bar x)\}
    \geq \widehat B(x)
    \geq \E_\alpha[B(\alpha)],
\]
so the better of $x$ and $\bar x$ is a good point for $B$. Constructing the symmetrized auxiliary functions and
choosing between the two candidates costs $O(n+m)$ additional work and
$O(\log m)$ depth, preserving the bounds of the theorem.
\end{remark}

To prove Theorem \ref{thm:monomial-linear}, we first normalize the problem. Writing
$\Psi_p(y) = \tfrac{\Psi_p(0)+\Psi_p(1)}{2} + \beta_p (-1)^y$ with
$\beta_p := \tfrac{\Psi_p(0)-\Psi_p(1)}{2}$, we obtain
\[
    B(x) = \E_\alpha[B(\alpha)] +
    \sum_{p = \{u,v\} \in \mathcal P} \beta_p (-1)^{x_u \oplus x_v}.
\]
A good point is
therefore a point at which the weighted sum 
$\sum_p \beta_p (-1)^{x_u \oplus x_v}$ is nonnegative.

\subsection{Fixing a seed instead of the variables}\label{subsec:seed-fixing}
Following Luby, we do not choose the bits $x_u$ directly but draw them
from a small structured space and apply conditional expectations to the
seed of that space, which has only $O(\log n)$ bits. Set
$\ell := \lceil \log_2 (n+1) \rceil$ and regard each
index $u \in \{1,\dots,n\}$ as an $\ell$-bit string. For a seed $r \in \{0,1\}^{\ell}$, define
\[
    x_u(r) := \bigoplus_{t=1}^{\ell} u_t\, r_t .
\]

This is the pairwise independent space given in Luby's original paper
\cite{Luby1993RemovingRandomness}. For a constraint $p = \{u,v\}$,
\[
    x_u(r) \oplus x_v(r) = \bigoplus_{t=1}^{\ell}  (d_p)_t\, r_t,
    \qquad d_p := u \oplus v \neq 0,
\]
so the value of the constraint is the sum mod two of the seed
coordinates at the positions where $d_p$ has a set bit, and at least
one such position exists. Suppose some such position $t$ is unfixed,
and write the value of the constraint as $r_t \oplus b$, where $b$ is the sum mod two of the seed coordinates at the remaining set
bits of $d_p$. Conditioned on arbitrary values of every coordinate
other than $r_t$, the bit $b$ is determined while $r_t$ is uniform, so
the value of the entire constraint is uniform as long as any nonzero bit of $d_p$ is unfixed. In particular, under a fully
uniform seed every constraint is a uniform bit and $\E_r[B(x(r))] = \E_\alpha[B(\alpha)]$ follows
by linearity of expectation.

The following structure essentially observed by Han in a less general form \cite{han1995improvement}, is the only new input to Luby's original strategy. Define the \emph{stage} of $p$ as
$t(p) := \max\{t : (d_p)_t = 1\}$, the most significant set bit of
$d_p$, and for $0 \leq t \leq \ell$ let
\[
    \Phi_t := \E_r\big[B(x(r)) \big| r_1, \dots, r_t\big]
\]
be the expectation over a uniform choice of the remaining seed
coordinates.

\begin{lemma}\label{lem:stage}
    Fix the seed coordinates in the order $r_1, r_2, \dots, r_\ell$.
    Then $\Phi_0 = \E[B(\alpha)]$, and for each stage $t$,
    \[
        \Phi_t - \Phi_{t-1} = (-1)^{r_t}\, F_t,
        \qquad F_t := \sum_{p\,:\,t(p) = t} \beta_p\, (-1)^{s_p},
        \qquad s_p := \bigoplus_{t' < t(p)} (d_p)_{t'}\, r_{t'} .
    \]
\end{lemma}
\begin{proof}
    A constraint with $t(p) > t$ has the unfixed coordinate $t(p)$ in
    its support, so its conditional expectation after stage $t$ is the
    constant $\tfrac{\Psi_p(0)+\Psi_p(1)}{2}$, regardless of the
    choices made; in particular
    $\Phi_0 = \sum_p \tfrac{\Psi_p(0)+\Psi_p(1)}{2} = \E[B(\alpha)]$. A
    constraint with $t(p) \leq t$ has all of its support fixed, and its
    value is $\Psi_p(s_p \oplus r_{t(p)})$. Having fixed the first $t-1$ bits of the seed, setting bit $t$ of the seed changes only the contribution of the
    constraints with $t(p) = t$. Accounting for the normalization, the change is thus
    \[
        \sum_{t(p)=t} \Big( \Psi_p(s_p \oplus r_t)
        - \tfrac{\Psi_p(0)+\Psi_p(1)}{2} \Big)
        = \sum_{t(p)=t} \beta_p (-1)^{s_p \oplus r_t}
        = (-1)^{r_t} F_t . \qedhere
    \]
\end{proof}

Lemma~\ref{lem:stage} determines the algorithm. Group the
constraints by stage; then for $t = 1, \dots, \ell$ in order,
compute $F_t$ by a parallel sum over the group of stage $t$, and set
$r_t := 0$ if $F_t \geq 0$ and $r_t := 1$ otherwise, so that
$\Phi_t - \Phi_{t-1} = |F_t|$. After the last stage every constraint is
determined, so
\[
    B(x(r)) = \Phi_\ell = \E_\alpha[B(\alpha)] + \sum_{t=1}^{\ell} |F_t|
    \geq \E_\alpha[B(\alpha)],
\]
and $x(r)$ is a good point. Each constraint participates in exactly one
stage, which is the source of the linear work bound. It remains to
implement the grouping and the quantities $s_p$ within it.

\subsection{Implementation in linear work}\label{subsec:seed-implementation} Each index, each
$d_p$, and the seed occupy single words. The algorithm uses two lookup tables, indexed by the $2^{\ell} \leq 2n$ possible words and computed before the stages begin. The table $\MSB$ stores in entry $w \geq 1$ the position of the most significant set bit of $w$, and
the table $\Par$ stores in entry $w$ the parity of the bits of $w$.

The table $\MSB$ is computed in $O(n)$ work and $O(1)$ depth: in
parallel over $t \in \{1, \dots, \ell\}$ and over the entries
$w \in [2^{t-1}, 2^{t})$, entry $w$ is set to $t$, and since these
ranges partition the index set, every entry is written exactly once.
The table $\Par$ is computed by induction on the word length. Write
$\Par_j$ for the table over $j$-bit words. The table $\Par_0$ consists
of the single entry $\Par[0] = 0$, and $\Par_{j+1}$ agrees with
$\Par_j$ on its first half and satisfies
$\Par[w + 2^{j}] = 1 \oplus \Par[w]$ on its second half, so it is
obtained from $\Par_j$ by one parallel copy and complement. The
$\ell$ rounds cost $O(2^{\ell}) = O(n)$ total work and $O(\log n)$
depth.

Given the tables, the quantities $d_p = u \oplus v$ and
$t(p) = \MSB[d_p]$ are computed in constant work per constraint. At
stage $t(p)$, let $r$ denote the word whose bits at the positions
fixed so far are the chosen seed coordinates, and let $\wedge$ denote
bitwise conjunction. The conjunction
$d_p \wedge r \wedge (2^{\,t(p)-1} - 1)$ retains exactly the bits
$(d_p)_{t'}\, r_{t'}$ for $t' < t(p)$, since the mask
$2^{\,t(p)-1} - 1$ removes the decisive position and every higher
position. Therefore
\[
    s_p = \Par\big[\, d_p \wedge r \wedge (2^{\,t(p)-1} - 1) \,\big],
\]
and the evaluation costs constant work per constraint.

The grouping of the constraints by decisive stage is a stable integer
sort with keys $t(p) \in \{1, \dots, \ell\}$. By
Lemma~\ref{lem:small-range-sort} with key range
$R = \ell = O(\log n)$, the sort costs $O(n+m)$ work and
$O(\log n + \log m)$ depth, and it places the constraints of each
stage in a contiguous segment of a single array. The value $F_t$ is
then computed at stage $t$ by one parallel sum over the segment of
stage $t$, in $O(\log m)$ depth.

%% file: short-note/graph-sparsification.tex
\section{Sparse Partitioning}\label{sec:sparsification}

In this section, we consider the task of partitioning the vertices of a graph
$G$ with $m$ edges into $k$ pieces $V_1,\dots,V_k$ such that the number of
edges among the induced subgraphs $G[V_1],\dots,G[V_k]$ sums to at most $m/L$.
We call such a partition an \emph{$L$-sparse $k$-partition}. For our
applications in Section~\ref{sec:derandom}, it will be convenient to define a
particular representation of such a partition. Call a neighbor $v$ of a vertex
$u \in V_i$ \emph{left}, \emph{internal}, or \emph{right} according to whether
$v$ lies in $V_1 \cup \dots \cup V_{i-1}$, in $V_i$, or in
$V_{i+1} \cup \dots \cup V_k$, and call $u$ \emph{isolated} if it has no internal
neighbors. The partition is \emph{ordered} if the adjacency array of every
vertex lists its left neighbors first, then its internal ones, then its right
ones, with the two block boundaries marked, and each piece $V_i$ stores the
list $I_i$ of its isolated vertices.

The main technical lemma is a linear deterministic procedure to construct a
single vertex partition.

\begin{lemma}\label{lem:two-graph-sparsification}
    The vertices of any graph $G$ may be divided into two parts so that the
    induced subgraphs $H_1$ and $H_2$ jointly contain at most $m/2$ edges,
    deterministically in $O(n+m)$ work and $O(\poly\log n)$ depth.
\end{lemma}

Iterating this routine, we obtain the following construction of sparse
partitions.

\begin{theorem}\label{cor:aggregate-sparsification}
    For any graph and any integer $k \geq 1$, an ordered $k$-sparse
    $O(k)$-partition can be computed deterministically in $O(n+m+k)$ work and
    $O(\log k \poly\log n)$ depth.
\end{theorem}

\begin{proof}
    We compute the partition by a recursion of
    $t := \lceil \log_2 k \rceil$ levels, whose parts form a binary tree rooted
    at the whole vertex set. At each level, every part containing at least one
    edge is split by an application of
    Lemma~\ref{lem:two-graph-sparsification} into a left and a right child,
    whose induced subgraphs jointly contain at most half the edges of the
    parent, while edgeless parts remain leaves. The parts after level $j$
    therefore contain at most $2^{-j}m$ edges in total, which is at most $m/k$
    after $t$ levels, and the number of parts is at most $2^t = O(k)$. The
    output parts $V_1,\dots,V_P$ are the leaves of the tree in left-to-right
    order.

    The recursion maintains two invariants. Each part holds its vertices in
    two lists, the \emph{active} vertices, which have a neighbor inside the
    part, and the isolated vertices, which have none. Furthermore, the
    partition is ordered, that is, the adjacency array of every vertex lists
    first its neighbors in parts to the left of its own, then its internal
    neighbors, and finally its neighbors in parts to the right, with two
    indices, the left and the right delimiter, marking the ends of the
    internal range. Initially, the whole vertex set is one part and every
    neighbor is internal, and at termination the second invariant is precisely
    the desired representation in the theorem.

    A part with an edge is split as follows. Let $A$ be its set of active
    vertices. The subgraph $G[A]$ is bisected by
    Lemma~\ref{lem:two-graph-sparsification} into the vertex sets $A_1$ and
    $A_2$ of the left and the right child, and the isolated list of the parent
    is passed to the isolated list of the left child by a pointer. Every
    vertex of $A$ records its side in a global array $\mathit{side}$, with one
    entry per vertex of $G$, allocated at the beginning of the procedure. The
    internal range of each vertex of $A$ is then partitioned using a filter
    based on the side array, and the delimiters are updated so that the second
    invariant is maintained. Finally, for each child, the newly isolated
    vertices are extracted from the current vertices by a filter and added to
    their respective isolated lists.

    An entry of an adjacency array is touched only while the endpoints of its
    edge share a part, and a vertex is never read after the split at which it
    becomes isolated. A part with $m'$ edges has $O(m')$ active vertices,
    so the work of level $j$ is $O(2^{-j}m)$ apart from a constant for passing
    down pointers. The total work is thus
    \[
        O(n+m) + \sum_{j=0}^{t-1} O(2^{-j}m) + O(2^t)
        = O(n+m+k).
    \]
    Each level costs $O(\poly\log n)$ depth by
    Lemma~\ref{lem:two-graph-sparsification}, so the depth is $O(\log k\poly\log n)$.
\end{proof}

Our algorithm to compute the desired partition of
Lemma~\ref{lem:two-graph-sparsification} instantiates a bit pair benefit
problem corresponding to computing a random cut. Allocating a bit variable $x_i$ for each vertex $u_i$, the corresponding
number of cut edges under an assignment $x$ is
\[
    B(x) := \sum_{\{u_i,u_j\} \in G} \ind[x_i \neq x_j]
    = \sum_{\{u_i,u_j\} \in G} x_i \oplus x_j.
\]
Each term has expectation $1/2$ under a uniform assignment, so
$\E[B(x)] = m/2$. Thus, a good assignment for the benefit problem can be
found by Theorem~\ref{thm:monomial-linear}. The partition of vertices
induced by this assignment has at least $m/2$ cut edges, and hence its
two induced subgraphs jointly contain at most $m/2$ edges, as desired.

For the complexity, the construction of the benefit instance and the final partition are
single scans of the vertex and edge lists, and the benefit problem costs
$O(n+m)$ work and $O(\log m\log n)$ depth by
Theorem~\ref{thm:monomial-linear}, giving the stated bounds.

%% file: short-note/deterministic-graph-algos.tex
\section{Derandomizing maximal independent set and graph coloring}\label{sec:derandom}
We now give deterministic algorithms that compute a maximal independent set and a $({\deg}+1)$-coloring in linear work. Both follow the same scheme. Theorem~\ref{cor:aggregate-sparsification}
first partitions the graph into $L$ pieces inducing $O(m/L)$ edges in total, and the pieces are then processed from left to right. Before a piece is solved, the partial solution on
the preceding pieces are extended: for MIS, vertices with a neighbor in the current independent set are deleted, and for coloring, the palettes of each vertex are updated by removing colors already used by its neighbors. This leaves a self-contained instance on the piece, which one of Luby's deterministic algorithms \cite{Luby1993RemovingRandomness} solves with work optimal up to a polylogarithmic factor, after which the next piece is processed in the same way. Taking $L$ a sufficiently large power of $\log n$ makes the subroutine calls cost $O(m)$ in total, while the updates cost $O(n+m)$, as they charge each vertex only its degree. Theorem~\ref{thm:extendable-linear} states the scheme abstractly, and the two corollaries following it verify its hypotheses.

One point requires care. The partition bounds the number of edges inside the pieces but not the number of vertices. Thus, naively applying Luby's algorithm to each piece of the partition can cost $\Omega(n \poly\log n)$ work. Hypothesis~\ref{item:extend} therefore requires that an isolated vertex of an instance be resolvable in constant work. Deleting the isolated vertices leaves an instance with at most twice as many vertices as edges, so the cost of each call to the subroutine is governed by the number of edges of its piece, and the sparsity of the partition makes the total linear.

\begin{theorem}\label{thm:extendable-linear}
    Let $\Pi$ be a graph problem with the following two properties.
    \begin{enumerate}[label=(\Alph*)]
        \item\label{item:extend} There is a deterministic procedure
        which, given disjoint $U, W \subseteq V(G)$, a solution $s$ of
        $\Pi$ on $G[U]$, and access to the cut $(U, W)$ and to
        $G[W]$, outputs an instance of $\Pi$ on a
        subgraph of $G[W]$ of size $O\big(|W| + |E(G[W])|\big)$, such that any solution of the
        instance combines with $s$ to give a solution of $\Pi$ on
        $G[U \cup W]$. The procedure performs
        $O\big(\sum_{v \in W}(1 + \deg_G(v))\big)$ work in
        $O(\poly\log n)$ depth, and an isolated vertex of the instance
        can be resolved and deleted in $O(1)$ work.
        \item\label{item:subroutine} $\Pi$ admits a deterministic
        algorithm solving any instance with $n'$ vertices, $m'$ edges,
        and size $O(n'+m')$ in
        $O\big((n'+m')\log^{a} n\big)$ work and $O(\poly\log n)$ depth,
        for some constant $a$.
    \end{enumerate}
    Then $\Pi$ can be solved deterministically in $O(n+m)$ work and
    $O(\poly\log n)$ depth.
\end{theorem}
\begin{proof}
    Set $L := \log^{a} n$ and let $V_1, \dots, V_P$ be the partition given by Theorem~\ref{cor:aggregate-sparsification} for $k = L$, computed in $O(n+m)$ work and $O(\poly\log n)$ depth, so that $P = O(\log^{a} n)$ and $\sum_i m_i \leq m/L$, where
    $m_i := |E(G[V_i])|$. Write $V_{<i} := V_1 \cup \dots \cup V_{i-1}$. The partition is ordered,
    so the adjacency array of each $v \in V_i$ lists its neighbors in $V_{<i}$ and then those in $V_i$, providing the access to the cut and to $G[V_i]$ that~\ref{item:extend} assumes for the pair $U = V_{<i}$, $W = V_i$.

    Process the parts in sequence, maintaining a solution $s_{i-1}$ of $\Pi$ on $G[V_{<i}]$. At step $i$, apply~\ref{item:extend} with
    $U = V_{<i}$ and $W = V_i$ to obtain an instance $T_i$ on a subgraph of $G[V_i]$. Each piece $V_i$ has its isolated vertex list, so those vertices of $T_i$ can be resolved and deleted immediately. This leaves the instance $T_i'$ given by the remaining vertices. We
    run the algorithm of~\ref{item:subroutine} on $T_i'$, and combine its output with $s_{i-1}$ to form $s_i$.
    By~\ref{item:extend} and induction on $i$, each $s_i$ solves $\Pi$ on $G[V_1 \cup \dots \cup V_i]$, so $s_P$ solves $\Pi$ on $G$.

    Every edge of $G$ is inspected at most twice, once from the part of each endpoint, so the extension steps cost
    \[
        \sum_{i=1}^{P} O\Big(\sum_{v \in V_i}
        \big(1 + \deg_G(v)\big)\Big) = O(n+m)
    \]
    work in total. Resolving and deleting isolated vertices costs $O(n)$ work in total. As $T_i'$ has at most $m_i$ edges and every vertex not originally in the isolated vertex list is incident to one of those $m_i$ edges, this instance has at most $2m_i$ vertices, so the call
    to~\ref{item:subroutine} costs $O(m_i \log^{a} n)$ work, and these sum
    to $O\big(\log^{a} n \cdot m/\log^{a} n\big) = O(m)$. Including the
    partition itself, the work is $O(n+m)$. The parts are processed
    sequentially and each costs $O(\poly\log n)$ depth, so the depth is
    $O(\poly\log n)$ as $P = O(\poly\log n)$.
\end{proof}
\begin{corollary}\label{cor:mis-linear}
    A maximal independent set of any graph can be computed deterministically in $O(n+m)$ work and $O(\poly\log n)$ depth.
\end{corollary}
\begin{proof}
    The independent set is maintained globally as an indicator array. We verify~\ref{item:extend}. Given a maximal independent set $M$ of $G[U]$, each $v \in W$ forms an auxiliary array over its left neighbors, filled by the respective neighbors indicator values. A parallel reduction over the auxiliary array decides whether $v$ has a neighbor in $M$. We then delete all such vertices, and output $H := G[W \setminus D]$. All of this performs $O(1 + \deg_G(v))$ work at each $v$ in $O(\log n)$ depth, and the size of $H$ is $O\big(|W| + |E(G[W])|\big)$. If $M'$ is a maximal independent set of $H$, then $M \cup M'$, formed by setting the bits of $M'$ in the indicator, is a maximal independent set of $G[U \cup W]$; any new incidence would have to come from a cut edge, and all such vertices were removed from consideration in the update step. An isolated vertex of $H$ is resolved by placing it in the independent set, which is sound because it has a neighbor neither in $M$ nor in $H$, and two isolated vertices are never adjacent. Property~\ref{item:subroutine} holds by any of the deterministic MIS algorithms of \cite{Luby1993RemovingRandomness,GoldbergSpencer1989Constructing,Han1996Derandomization}, so Theorem~\ref{thm:extendable-linear} applies.
\end{proof}

\begin{corollary}\label{cor:coloring-linear}
    A $({\deg}+1)$-coloring of any graph can be computed deterministically in $O(n+m)$ work and $O(\poly\log n)$ depth. In particular, so can a $(\Delta+1)$-coloring.
\end{corollary}
\begin{proof}
    We verify~\ref{item:extend}. Given a proper coloring $\varphi$ of $G[U]$ from its palettes, each $v \in W$ computes its residual palette $C'(v) := C(v) \setminus \{\varphi(u) : u \in N_G(v) \cap U\}$ as follows: in a Boolean array indexed by $\{1, \dots, \deg_G(v)+1\}$, each adjacency entry of $v$ holding a neighbor $u \in U$ with $\varphi(u) \leq \deg_G(v)+1$ marks the cell $\varphi(u)$, and $v$ then stably filters $C(v)$, keeping the unmarked colors. Colors above $\deg_G(v)+1$ never lie in $C(v)$, so $C'(v)$ is as claimed, and stability ensures the palette remains sorted. This costs $O(\deg_G(v)+1)$ work at each $v$ and $O(\log n)$ depth. At most $\deg_U(v)$ colors are removed, so
    \[
        |C'(v)| \geq \deg_G(v) + 1 - \deg_U(v) \geq \deg_{G[W]}(v) + 1.
    \]
    Each palette $C'(v)$ is then truncated to size $\deg(v)+1$. The resulting graph $G[W]$ with the palettes $C'$ is again an instance of $(\deg + 1)$-coloring. No vertex of $W$ retains a color used by a neighbor in $U$, so any proper coloring from the palettes $C'$ combines with $\varphi$ into one of $G[U \cup W]$, and an isolated vertex is resolved by the first color of its palette, which is nonempty. For~\ref{item:subroutine}, the subroutine receives an instance of the above form with $n'$ vertices and $m'$ edges, after which we apply the deterministic coloring algorithm from Luby's paper \cite{Luby1993RemovingRandomness}. We remark here that Luby's coloring algorithm is presented as only using $\Delta+1$ colors, but his exact algorithm works given that each vertex $v$ has a sorted palette $C(v)$ of size $\deg(v)+1$, so we may apply it to our instance.
\end{proof}